\documentclass[a4paper,11pt]{paper}

\usepackage{fullpage,latexsym,amsthm,amsmath,amssymb,url}
\usepackage[ruled, linesnumbered]{algorithm2e}
\usepackage{tikz}
\usepackage{wasysym,marvosym}
\usepackage{comment}

\usepackage{enumerate}
\usepackage[shortlabels]{enumitem}
\usepackage{hyperref}
\usetikzlibrary{arrows}

\usepackage{boxedminipage}
\usepackage{marvosym}
\usepackage{stmaryrd}
\usepackage[all]{xy}
\usepackage[normalem]{ulem}
\usepackage{soul}

\input colordvi       

\newtheorem{lemma}{Lemma}

\newtheorem{observation}{Observation}

\newtheorem{definition}{Definition}
\newtheorem{theorem}{Theorem}

\newcommand{\cw}{\mathsf{ctw}}

\newcommand{\domleq}{\preceq_{\mathrm{dom}}}
\newcommand{\cutvector}[2]{\mathsf{cuts}\langle #1,#2\rangle}

\newcommand{\fstsmt}[2]{#1_{\leq #2}}

\newcommand{\lstsmt}[2]{#1_{> #2}}

\newcommand{\onesmt}[2]{#1_{#2}}

\newcommand{\permut}[1]{\frak{S}({#1})}

\title{Strong immersion is a well-quasi-ordering for semi-complete digraphs\thanks{The research of F. Barbero and C. Paul is supported by the DE-MO-GRAPH project ANR-16-CE40-0028. 
The research of Mi. Pilipczuk is supported by Polish National Science Centre grant no. 2013/11/D/ST6/03073. Mi. Pilipczuk is also supported by the Foundation for Polish Science via the START stipend programme.}}
\author{Florian Barbero \thanks{LIRMM, Universit\'e de Montpellier, France, \texttt{florian.barbero@lirmm.fr}} 
   \and Christophe Paul \thanks{LIRMM, CNRS, Universit\'e de Montpellier, France, \texttt{christophe.paul@lirmm.fr}} 
   \and Micha\l{}' Pilipczuk \thanks{University of Warsaw, Poland, \texttt{michal.pilipczuk@mimuw.edu.pl}} }

\date{\today}

\begin{document}

\maketitle

\begin{abstract}
We prove that the strong immersion order is a well-quasi-ordering on the class of semi-complete digraphs, 
thereby strengthening a result of Chudnovsky and Seymour~\cite{CS11} that this holds for the class of tournaments.
\end{abstract}



\section{Introduction}

Understanding the combinatorics of inclusion relations of graphs is a central question in algorithmic graph theory. 
In their celebrated series of papers (see~\cite{RS04}), Robertson and Seymour proved Wagner's conjecture~\cite{Wag37} stating that undirected graphs are well-quasi-ordered under the \emph{minor} relation. 
In other words, every minor-closed graph property can be characterized by a finite set of excluded minors. 
Together with the cubic algorithm that tests whether a fixed graph $H$ is a minor of an input graph $G$~\cite{RS95c}, 
the graph minors theorem led to the development of a deep algorithmic graph structure theory, particularly useful in the context of parameterized complexity; cf.~\cite{CFK14,DF13,FG06}.
Notably, the work on the theory of graph minors established the notions of treewidth and tree decompositions, which are now key structural concepts in the (parameterized) algorithm design.

Graph inclusion relations alternative to the minor relation also attracted a lot of interest. Among them, we may consider the recent developments on \emph{(weak) immersions}:
a graph $H$ can be {\em{weakly immersed}} in a graph $G$ if there is a mapping $\mu$ from vertices of $H$ to pairwise different vertices of $G$ and from edges of $H$ to pairwise edge-disjoint paths in $G$
such that for an edge $uv$ of $G$, the endpoints of $\mu(uv)$ are $\mu(u)$ and $\mu(v)$.
It is known that finite undirected graphs are well-quasi-ordered under weak immersion~\cite{RS10}. 
Testing whether an input undirected graph $G$ contains a graph $H$ as a weak immersion is fixed-parameterized tractable when parameterized by the size of $H$~\cite{GKM11}.
Very recently, a structural parameter {\em{tree-cut width}} was introduced and studied as an appropriate width measure for immersions in undirected graphs.
It seems that tree-cut width possesses a number of good combinatorial and algorithmic properties similar to those that made treewidth so successful; cf.~\cite{GPT16,KPST15,Wol15}.

Developing a similar theory of inclusion relations in the context of directed graphs is a long-standing research challenge. 
What is the right notion of inclusion relation or what is the appropriate width parameter are questions that have not yet received a clear answer. 
For example, the recent directed grid theorem~\cite{KK14} in the context of directed treewidth relies on the notion of 
a \emph{butterfly minor}\footnote{We do not need here the exact definition of a butterly minor hence we refrain from giving it.}. 
However, general digraphs are not well-quasi-ordered under the butterfly minor relation, 
as the set of oriented cycles depicted in Figure~\ref{fig:antichain} forms an infinite antichain in the butterfly minor relation.

\begin{figure}
\centerline{\includegraphics[width=8cm]{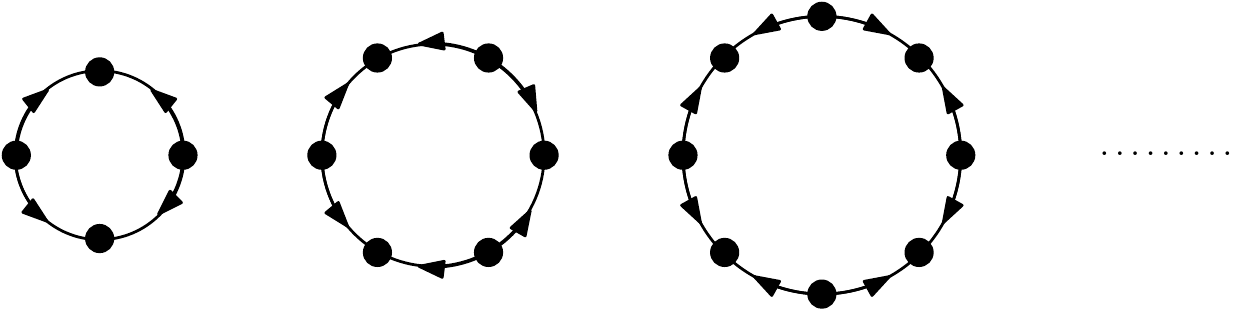}}
\caption{The set $\{C_{2k}\mid k=2,3,4,\ldots\}$, where $C_{2k}$ is a cycle on $2k$ vertices with arcs alternately oriented clockwise and counter-clockwise.\label{fig:antichain}}
\end{figure}

For this reason, special classes of digraphs were studied in the context of inclusion relations.
In this work we focus on the classes of {\em{tournaments}} and {\em{semi-complete digraphs}}.
A {\em{simple}} digraph is one where there are no loops nor multiple arcs with the same head and tail; however, we allow the presence of two arcs of the form $(u,v)$ and $(v,u)$ at the same time, which we call
{\em{symmetric arcs}}.
A simple digraph is {\em{semi-complete}} if for every pair of different vertices $u,v$, at least one of the arcs $(u,v)$ and $(v,u)$ is present in the digraph; it is a {\em{tournament}}, if exactly one
of them is present for every pair $u,v$. Thus, in tournaments we forbid symmetric arcs, while in semi-complete digraphs we allow them.

Kim and Seymour proved that the butterfly-minor relation is a well-quasi-ordering on the class semi-complete digraphs~\cite{Kim13}.
It is believed that this result cannot be generalized to natural larger classes of digraphs. 
Indeed, it is conjectured in~\cite{Kim13} that neither the class supertournaments nor the class simple digraphs with stability number at most two is well-quasi-ordered under the butterfly-minor relation.

For immersions, we may consider the (weak) immersion relation for directed graphs, defined similarly as the weak immersion relation for undirected graphs, but 
undirected paths are replaced with directed ones.
We may also consider the (strong) immersion relation where we additionally require that no path that is an image under $\mu$ of some arc $(u,v)$ traverses a vertex $\mu(w)$ for some $w\notin \{u,v\}$.
See Section~\ref{sec:prelim} for a formal definition.
The set depicted in Figure~\ref{fig:antichain} is again an infinite antichain for both the strong or the weak immersion relations in general digraphs.
However, as proved by Chudnovsky and Seymour~\cite{CS11}, strong immersion is a well-quasi-ordering on the class of tournaments.

\begin{theorem}[\cite{CS11}]\label{thm:wqoT}
Tournaments are well-quasi-ordered under strong immersions.
\end{theorem}

The proof of Theorem~\ref{thm:wqoT} of Chudnovsky and Seymour~\cite{CS11} actually does not work directly in the more general setting of semi-complete digraphs, 
despite this statement being circulated in the literature~\cite{KimS15,Pil13}. More precisely, there is a technical issue in one of the considered cases, 
where it is crucially used that the digraphs in question do not contain symmetric arcs. 
In this work we fill this gap by generalizing Theorem~\ref{thm:wqoT} to semi-complete digraphs.

\begin{theorem}\label{thm:wqoSC}
Semi-complete digraphs are well-quasi-ordered under strong immersions.
\end{theorem}

As for every digraph $H$ there is a cubic time algorithm that decides whether there is a strong or weak immersion of $H$ in an input semi-complete digraph $D$~\cite{CFS12},
Theorem~\ref{thm:wqoT} has a number of meta-algorithmic consequences, for instance a cubic algorithm for the recognition of any fixed immersion-closed class of tournaments.
Theorem~\ref{thm:wqoSC} allows us to extend these corollaries to semi-complete digraphs.

Let us briefly explain our approach to the proof of Theorem~\ref{thm:wqoSC}.
In a nutshell, we follow closely the approach of Chudnovsky and Seymour~\cite{CS11} and we ``patch'' the crucial step in the proof where the assumption about the non-existence of symmetric arcs is used.
This patch is not straightforward and requires some new combinatorial ideas.

The crux of the proof of Chudnovsky and Seymour~\cite{CS11} is to use a structural parameter {\em{cutwidth}} (we define it formally in Section~\ref{sec:prelim}), 
which is bound to strong immersions via the following result of Chudnovsky, Fradkin, and Seymour~\cite{CFS12}.

\begin{lemma}[\cite{CFS12}]\label{lem:exclusion}
Let $\mathcal{F}$ be a family of semi-complete digraphs. Then the following conditions are equivalent.
\begin{itemize}
\item There exists a positive integer $c$ such that every member of $\mathcal{F}$ has cutwidth at most $c$.
\item There exists a digraph $H$ such that $H$ cannot be strongly immersed in any member of $\mathcal{F}$.
\end{itemize}
\end{lemma}

Note that Lemma~\ref{lem:exclusion} works in the semi-complete setting.
By Lemma~\ref{lem:exclusion}, proving Theorem~\ref{thm:wqoSC} boils down to the following statement, as we explain next.

\begin{lemma}\label{lem:wqo}
For every nonnegative integer $c$, the strong immersion relation is a well-quasi-ordering on semi-complete digraphs of cutwidth at most $c$.
\end{lemma}

We now repeat the argument of Chudnovsky and Seymour~\cite{CS11} in order to show how Theorem~\ref{thm:wqoSC} can be obtained by combining Lemmas~\ref{lem:exclusion} and~\ref{lem:wqo}.
Take any infinite sequence $S_1,S_2,S_3,\ldots$ of semi-complete digraphs. 
It suffices to prove that there are some $1\leq i<j$ such that $S_i$ can be strongly immersed in $S_j$.
If $S_1$ can be strongly immersed in any $S_j$ for $j\geq 2$, then we are done, hence assume otherwise.
By Lemma~\ref{lem:exclusion}, there exists $c$ such that each of the semi-complete digraphs $S_2,S_3,S_4,\ldots$ has cutwidth at most $c$.
By Lemma~\ref{lem:wqo}, there are some $2\leq i<j$ such that $S_i$ can be strongly immersed in $S_j$, and we are done.

The proof of the counterpart of Lemma~\ref{lem:wqo} in~\cite{CS11} is essentially done by encoding a small-width layout of a tournament in a word over an alphabet of size dependent in $c$ in such a way
that a Higman-like embedding of words encoding two tournaments implies the existence of a strong immersion from one to the other.
The place where the non-existence of symmetric arcs is used in~\cite{CS11} lies in the proof of this implication.
In order to prove the more general statement of Lemma~\ref{lem:wqo}, we enrich the encoding of a small-width layout of a semi-complete digraph
by including also information about symmetric arcs. This allows us to find appropriately 
embedded paths for them as well. The technical details of this step rely on a good understanding of the proof of~\cite{CS11}, so we defer further explanation to Section~\ref{sec:proof}.




\section{Preliminaries}
\label{sec:prelim}

\paragraph{Basic definitions and notation.}
A relation $\preceq$ over a set $S$ is a \emph{quasi-ordering} if it is transitive and reflexive. 
We say that $\preceq$ is a {\em{well-quasi-ordering}} ({\sf WQO} for short)
if for every infinite sequence $x_1,x_2,\ldots$ of elements of $S$ there exist $1\leq i< j$ such that $x_i\preceq x_j$.
It is well-known that this is equivalent to saying that for every subset $T$ of $S$ closed under $\preceq$, there is a finite set $F\subseteq S$
such that an element $x$ of $S$ belongs to $T$ if and only if $f\not\preceq x$ for each $f\in F$.

We use standard graph notation for directed graph (digraph). For a digraph $D$, the vertex and arc sets of $D$ are denoted by $V(D)$ and $E(D)$, respectively. 
For an arc $(u,v)$ of a digraph $D$, vertex $u$ is called the tail and vertex $v$ the head. 
Given a bipartition $(A,B)$ of the vertex set $V(D)$, we define the \emph{cut} $E(A,B)$ as the subset of arcs $\{(u,v)\mid u\in A, v\in B\}$.

For two integers $p \leq p'$, let $[p,p'] \subseteq \mathbb{Z}$ be the set of integers between $p$ and $p'$. If $p < p'$, we set $[p',p] = \emptyset$ by convention. 
For a finite set $S$, by $\permut{S}$ we denote the set of permutations of the elements of $S$. A permutation $\sigma\in\permut{S}$ is seen as a bijective mapping $\sigma\colon S\rightarrow [1,|S|]$. An element $u \in S$ is \emph{at position $i$ in $\sigma$} if $\sigma(u) = i$, and we denote this unique element by $\onesmt{\sigma}{i}$. The {\em{prefix}} of length $i$ of $\sigma$ is the set $\fstsmt{\sigma}{i}= \{\onesmt{\sigma}{j} \colon j \in [1,i]\}$; we set $\fstsmt{\sigma}{i} = \emptyset$ when $i \leq 0$, and $\fstsmt{\sigma}{i} = S$ when $|S| \leq i$. We extend this notation to prefixes and suffixes of orderings naturally, e.g., $\lstsmt{\sigma}{i} = S\setminus \fstsmt{\sigma}{i}$ is the set of the last $n-i$ vertices in $\sigma$.

Let $D$ be a digraph. A permutation $\pi\in\permut{V(D)}$ is called a \emph{vertex ordering}. An arc $(\onesmt{\pi}{i},\onesmt{\pi}{j}) \in E(D)$ is a \emph{feedback arc for the vertex ordering $\pi$} if $i > j$, that is, if $\onesmt{\pi}{i}$ is after $\onesmt{\pi}{j}$ in $\pi$. The sequence of \emph{cuts} of $\pi$ is defined as $(E^0_\pi,\dots,E^n_\pi)$ where for $i\in[0,n]$, $E^i_\pi=E(\lstsmt{\pi}{i},\fstsmt{\pi}{i})$. Hereafter, a permutation $\varepsilon_{\pi}^i\in \permut{E_{\pi}^i}$ will be called an \emph{ordered cut} of $\pi$. Finally, we set $\cutvector{D}{\pi} = (|E^0_\pi|, |E^1_\pi|, \ldots, |E^n_\pi|)$, which can be interpreted the function such that $\cutvector{D}{\pi}(i)=|E^i_\pi|$.

\begin{definition}
Let $\pi$ be a vertex ordering of a digraph $D$. The {\em width} of $\pi$ is
$$\cw(D,\pi)=\max\{\cutvector{D}{\pi}\}$$
where $\max$ on a tuple yields the largest coordinate. 
The {\em{cutwidth}} of $D$ is
$$\cw(D)=\min\{\cw(D,\pi)\mid\mbox{$\pi$ is a vertex ordering of $D$}\}.$$
A vertex ordering $\pi$ of $D$ is {\em{$\cw$-optimal}} if $\cw(D,\pi)=\cw(D)$.
\end{definition}


\begin{definition}
Let $D$ and $H$ be two digraphs. A \emph{strong immersion} of $H$ in $D$ is a mapping $\mu$ such that:
\begin{enumerate}
\item $\mu$ maps $V(H)$ injectively to $V(D)$;
\item for every $(u,v)\in E(H)$, $\mu((u,v))$ is a directed path from $\mu(u)$ to $\mu(v)$ in $D$;
\item for every pair of distinct arcs $e,f\in E(H)$, the directed paths $\mu(e)$ and $\mu(f)$ are arc-disjoint;
\item for every arc $e\in E(H)$ and every vertex $v\in V(H)$ not incident to $e$, the vertex $\mu(v)$ does not lie on the directed path $\mu(e)$.
\end{enumerate}
\end{definition}

\begin{figure}
\centerline{\includegraphics[width=8cm]{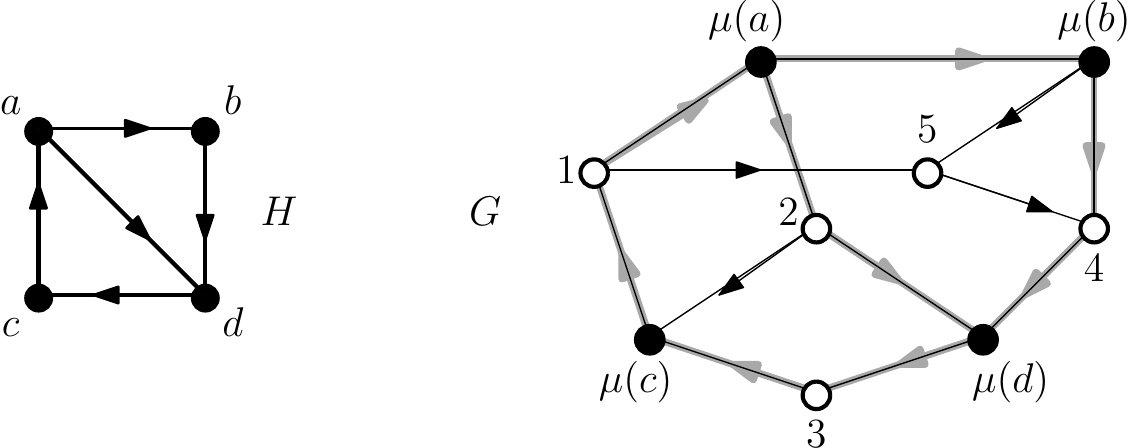}}
\caption{An immersion $\mu$ of the graph $H$ in $G$. The embedding of arcs of $H$ into paths of $G$ is depicted in grey lines. We have that $\mu((a,b))=[\mu(a),\mu(b)]$, $\mu((b,d))=[\mu(b),4,\mu(d)]$, $\mu((d,c))=[\mu(d),3,\mu(c)]$, $\mu((c,a))=[\mu(c),1,\mu(a)]$ and $\mu((a,d))=[\mu(a),2,\mu(d)]$}
\end{figure}

\paragraph{Linked vertex ordering and linked sequence of ordered cuts.}
We recall the definitions of the main tools used in~\cite{CS11} for the proof of the counterpart of Lemma~\ref{lem:wqo} for tournaments. 

\begin{definition}
Let $D$ be a digraph of order $n$. A vertex ordering $\pi$ of the vertex set $V(D)$ is a \emph{linked vertex ordering}\footnote{In~\cite{CS11}, the authors used the terminology \emph{linked enumeration}.} 
if for every $i,j\in [n]$ with $i<j$ such that $|E_\pi^i| = |E_\pi^j|=t$, the following holds:
\begin{itemize}
\item either there exists $h \in [i,j]$ such that $|E_\pi^h|< t$, or
\item there exist $t$ arc-disjoint paths from $\pi_{> j}$ to $\pi_{\leq i}$.
\end{itemize}
\end{definition}



The definition of linked vertex ordering  is extended to sequence of ordered cuts as follows.

\begin{definition}
Let $(E^0_\pi,\dots,E^n_\pi)$ be the sequence of cuts of  a linked vertex ordering $\pi$ of a digraph $D$. Then a sequence $(\varepsilon_{\pi}^0,\ldots,\varepsilon_{\pi}^n)$ of ordered cuts of $\pi$, with $\varepsilon_{\pi}^i\in\permut{E^i_\pi}$ for $i\in[0,n]$,  is \emph{linked} if for every $i,j$ with $0\leq i<j\leq n$ such that $|E_\pi^i| = |E_\pi^j|=t$, we have that:
\begin{itemize}
\item either there exists $h \in [i,j]$ such that $|E_\pi^h|< t$, or
\item there exist $t$ arc-disjoint paths $P_1,\dots, P_t$ from $\pi_{> j}$ to $\pi_{\leq i}$ such that for all $s\in [t]$, 
the path $P_s$ starts with the arc $\varepsilon_{\pi}^i(s)$ and ends with the arc $\varepsilon_{\pi}^j(s)$.
\end{itemize}
\end{definition}

As shown by Chudnovsky and Seymour~\cite{CS11}, there is always a $\cw$-optimal vertex ordering that is linked. Moreover, given a linked vertex ordering, one can construct a linked sequence of ordered cuts.

\begin{lemma}[3.1 of~\cite{CS11}]\label{lem:lvo}
Let $D$ be a digraph. Then there exists a linked vertex ordering $\pi$ of $D$ with $\cw(D,\pi)=\cw(D)$.
\end{lemma}

\begin{lemma}[5.1 of~\cite{CS11}]\label{lem:lco}
Let $\pi$ be a linked vertex ordering of a digraph $D$.
Then for each $i\in [0,n]$ there exists an ordered cut $\varepsilon_{\pi}^i\in\permut{E_{\pi}^i}$ so that 
$(\varepsilon_{\pi}^0,\dots \varepsilon_{\pi}^n)$ is a linked sequence of ordered cuts.
\end{lemma}

\paragraph{Codewords and domination relation.} 
Let $L$ be a finite set of labels and $c$ be a nonnegative integer.  
An \emph{$(L,c)$-codeword}\footnote{Observe that in~\cite{CS11}, the term \emph{codeword} is used differently: our definition of an $(L,c)$-codeword corresponds to an $(L,c)$-gap sequence in~\cite{CS11}.} is defined as a triple $(n,\lambda,\zeta)$ where $n$ is a positive integer, $\lambda \colon [n] \rightarrow L$, and $\zeta \colon [n-1] \rightarrow [0,c]$. 
The set of $(L,c)$-codewords is equipped with a partial order $\domleq$, called \emph{domination}, defined as follow. 
Given two $(L,c)$-codewords $(n,\lambda,\zeta)$ and $(n',\lambda',\zeta')$, we have $(n,\lambda,\zeta) \domleq (n',\lambda',\zeta')$ if and only if there exists
a strictly increasing function $f\colon [n]\to [n']$, called the {\em{embedding}}, such that
\begin{itemize}
\item for all $j \in [n]$, we have $\lambda(j) = \lambda'(f(j))$; and 
\item for all $j \in [n-1]$ and all $i \in [f(j),f(j+1)-1]$, we have $\zeta(j) \leq \zeta'(i)$.
\end{itemize}
Observe that a $(L,c)$-codeword can be seen as a directed path with $n$ vertices with labels on both vertices and arcs.
Vertices are labeled by $\lambda$ with labels from $L$, while arcs are labeled by $\zeta$ with integers from $[0,c]$; 
here, an argument $i\in [n-1]$ of $\zeta$ is intepreted as the arc from the $i$th to the $(i+1)$st vertex of the path.
As in~\cite{CS11}, we can use the variant of Higman's lemma due to \cite{Kri89,Sim85} to infer the following. Note here that $L$ is a finite set, so equality on it is a well-quasi-ordering.

\begin{lemma}\label{lem:WQO}
The set of $(L,c)$-codewords with the domination order is a well-quasi-ordering.
\end{lemma}

To prove that tournaments are well-quasi-ordered under strong immersions, Chudnovsky and Seymour~\cite{CS11} represent a tournament equipped with a linked vertex ordering using an $(L,c)$-codeword, for some finite $L$. They show that given two tournaments $T$ and $T'$, if a $(L,c)$-codeword representing $T$ is dominated by the one representing $T'$, then the domination relation allows to reconstruct the immersion of $T$ in $T'$. We explain this formally in the next section.
%
%
%


\section{Semi-complete digraphs are well-quasi-ordered under strong immersion}
\label{sec:proof}



This section is devoted to a proof of Lemma~\ref{lem:wqo} restated below:

\setcounter{lemma}{1}
\begin{lemma}
For every nonnegative integer $c$, the strong immersion relations is a well-quasi-ordering on semi-complete digraphs of cutwidth at most $c$.
\end{lemma}
\setcounter{lemma}{5}

The idea of the proof is as follows. 
We first define $(L,c)$-codewords, for some finite $L$, in order to represent the structure of a semi-complete digraph equipped with a linked vertex ordering. 
The way we construct $(L,c)$-codewords will extend the $(L,c)$-codewords defined by Chudnowsky and Seymour~\cite{CS11} to represent tournaments.
Given two semi-complete digraphs $S$ and $S'$, immersion of $S$ in $S'$ will be constructed from the domination relation between the respective $(L,c)$-codewords. 
We proceed in two steps. First we partition the arc set $E(S)$ of $S$ into $E_1$ and $E_2$ such that $T=(V(S),E_1)$ forms a tournament. 
Then we apply Chudnovsky and Seymour~\cite{CS11} to build an immersion model $\mu$ of $T$ in $S$. 
It then remains to extend $\mu$ to an immersion model of $S$ in $S'$ by mapping the arcs of $E_2$ to ``free'' paths of $S'$.

\newcommand{\codeword}[2]{\mathsf{code}\langle #1,#2\rangle}
\newcommand{\profile}[2]{\mathsf{profile}\langle #1,#2\rangle}
\newcommand{\symmetric}[2]{\mathsf{symmetric}\langle #1,#2\rangle}

\medskip
Let us first describe how a semi-complete digraph $S$, such that $\cw(S)\leq c$ for some nonnegative integer $c$, can be represented by an $(L,c)$-codeword, for some finite $L$.
Let $\pi$ be a vertex ordering of  $S$ such that  $\cw(S,\pi)\leq c$.
Fix any sequence of ordered cuts $\sigma=(\varepsilon_{\pi}^0,\dots \varepsilon_{\pi}^n)$ of $\pi$.
Based on these, we will define the \emph{encoding of $S$ with respect to $(\pi,\sigma)$} as 
$$(n,\codeword{S}{\pi,\sigma},\cutvector{S}{\pi}).$$

Recall here that $\cutvector{S}{\pi}=(|E_\pi^0|,|E_\pi^1|,\ldots,|E_\pi^{n}|)$, hence we need to define $\codeword{S}{\pi,\sigma}$.
Intuitively, $\codeword{S}{\pi,\sigma}$ is a function tailored to representing, for every $i\in[n]$, 
the structure of consecutive ordered cuts $(\varepsilon_{\pi}^{i-1},\varepsilon_{\pi}^{i})$ of $\pi$ using only a finite set of labels. 
To that aim, we define an equivalence relation $\sim$ over pairs of ordered cuts of size at most $c$. Let $(E_1,E_2)$ and $(F_1,F_2)$ be two pairs of cuts of size at most $c$ and consider
ordered cuts $\varepsilon_1\in\permut{E_1}$,  $\varepsilon_2\in\permut{E_2}$,  $\varphi_1\in\permut{F_1}$,  $\varphi_2\in\permut{F_2}$. 
Then $(\varepsilon_1,\varepsilon_2)\sim (\varphi_1,\varphi_2)$ if and only if the following conditions
are satisfied:
\begin{itemize}
\item $|E_1|=|F_1|$ and  $|E_2|=|F_2|$; and 
\item $\varepsilon_1(i)=\varepsilon_2(j)$ if and only if $\varphi_1(i)=\varphi_2(j)$, for all relevant indices $i,j$. 
\end{itemize}
 
 \noindent
To complete the description of the function $\codeword{S}{\pi,\sigma}$, for every $i\in [n]$, we define:
$$\codeword{S}{\pi,\sigma}(i)=(\ \profile{S}{\pi,\sigma}(i)\ ,\ \symmetric{S}{\pi,\sigma}(i)\ ,\ i \bmod (4c+1)\ ),$$ 
where
\begin{itemize}
\item $\profile{S}{\pi,\sigma}(i)$ is the equivalence class of $(\varepsilon_{\pi}^{i-1},\varepsilon_{\pi}^{i})$ with respect to $\sim$; and
\item $\symmetric{S}{\pi,\sigma}(i)\colon E_{\pi}^{i-1}\cup E_{\pi}^{i}\rightarrow \{\top,\bot\}$ is a function distinguishing symmetric arcs from others.
That is, for $(u,v)\in E_{\pi}^{i-1}\cup E_{\pi}^{i}$ we have $\symmetric{S}{\pi,\sigma}(i)(u,v) = \top$ if and only if we also have $(v,u) \in E(S)$.
\end{itemize}
We remark that the information encoded in a codeword by Chudnovsky and Seymour in~\cite{CS11} is exactly $\profile{S}{\pi,\sigma}$ and $\cutvector{S}{\pi}$.
Here, we extend this information by two components: $\symmetric{S}{\pi,\sigma}$ that stores information on symmetric arcs, and the remainder of the index modulo $4c+1$.
The latter will be used for a technical reason in the proof.

Observe that, as $\cw(S,\pi)\leq c$, the size of every cut of $\pi$ is at most $c$ and the number of equivalence classes of $\sim$ is bounded by a function of $c$.
It follows that all the values of $\profile{S}{\pi,\sigma}(i)$ and of $\codeword{S}{\pi,\sigma}(i)$, for $i\in [0,n]$, belong to some finite
sets whose sizes are bounded by a function of $c$. Let us denote by $L^c_{\mathsf{profile}}$ and $L^c_{\mathsf{code}}$ these (finite) sets of possible values, respectively.
It follows that:



\begin{observation}
Let $S$ be a semi-complete digraph with $\cw(S)\leq c$. 
Suppose $\sigma=(\varepsilon_{\pi}^0,\dots \varepsilon_{\pi}^n)$ is a sequence of ordered cuts of a $\cw$-optimal vertex ordering $\pi$ of $S$.
Then $(n,\profile{S}{\pi,\sigma},\cutvector{S}{\pi})$ is an $(L^c_{\mathsf{profile}},c)$-codeword and $(n,\codeword{S}{\pi,\sigma},\cutvector{S}{\pi})$ is an $(L^c_{\mathsf{code}},c)$-codeword.
\end{observation}

It essentially remains to prove that a strong immersion between two semi-complete digraphs of cutwidth at most $c$ can be inferred from the domination relation between two $(L,c)$-codewords associated with them.
In~\cite{CS11}, this claim for the tournament case is proved in (5.2). The following Lemma~\ref{lem:dom-tour} is a rephrasing of this result, with some additional assertions added.
These assertions follow from a straightforward inspection of the proof of (5.2) in~\cite{CS11}; let us review them quickly.
\begin{itemize}
\item We do not need to assume that the vertex ordering of the embedded tournament is linked. Similarly for the sequence of ordered cuts associated with it.
\item We do not need to assume that the host digraph $S$ is a tournament; semi-completeness of $S$ suffices. 
\item The additional properties of the constructed immersion model $\mu$ follow directly from the construction.
\end{itemize}

\begin{lemma}[\cite{CS11}]\label{lem:dom-tour} 
Let $\pi$ be a vertex ordering of a tournament $T$ on $n$ vertices such that $\cw(T,\pi)\leq c$ and $\pi'$ be a linked vertex ordering of a semi-complete digraph $S'$ on $n'$ vertices 
such that $\cw(S,\pi')\leq c$.
Suppose further that $\sigma$ and $\sigma'$ are sequences of ordered cuts for $\pi$ and $\pi'$, respectively, where $\sigma'$ is linked.
Finally, suppose that
$$(n,\profile{T}{\pi,\sigma},\cutvector{T}{\pi}) \domleq (n',\profile{S'}{\pi',\sigma'},\cutvector{S'}{\pi'}),$$
which is certified by an embedding $f\colon [n]\to [n']$.
Then there exists a strong immersion model $\mu$ of $T$ in $S$ where $\mu(\pi_j)=\pi'_{f(j)}$ for all $j\in [n]$.
Moreover, for every arc $e$ of $S$, 
\begin{itemize}
\item \emph{[(3) in proof of (5.2) in \cite{CS11}]} if $e$ is a feedback arc in $\pi$, then the path $\mu(e)$ both starts and ends with a feedback arc in $\pi'$,
\item \emph{[(5) in proof of (5.2) in \cite{CS11}]} if $e$ is not a feedback arc in $\pi$, then the path $\mu(e)$ consists of one arc.
\end{itemize}
\end{lemma}

We are now ready to generalize Lemma~\ref{lem:dom-tour} to the setting of semi-complete digraphs.

\begin{lemma}\label{lem:dom-semi}
Let $S=(V,E)$ and $S'=(V',E')$ be two semi-complete digraphs with associated vertex orderings $\pi$ and $\pi'$, respectively, where $\pi'$ is linked and such that $\cw(S,\pi)\leq c$ and $\cw(S',\pi')\leq c$.
Suppose further that $\sigma$ and $\sigma'$ are sequences of ordered cuts for $\pi$ and $\pi'$, respectively, where $\sigma'$ is linked.
Finally, suppose that
$$(n,\codeword{S}{\pi,\sigma},\cutvector{S}{\pi})\domleq (n',\codeword{S'}{\pi',\sigma'},\cutvector{S'}{\pi'}).$$ 
Then $S$ can be strongly immersed in $S'$.
\end{lemma}
\begin{proof} 
We adopt the notation from the definition of domination: we have an embedding $f\colon [n]\to [n']$ such that for all $j \in [n]$ we have $\codeword{S}{\pi,\sigma}(j) = \codeword{S'}{\pi',\sigma'}(f(j))$, and
for all $j \in [n-1]$ and all $i \in [f(j),f(j+1)-1]$ we have $|E_{\pi}^j| \leq |{E'}_{\pi'}^{i}|$.

We partition $E$ into $E_1\uplus E_2$ so that $T=(V,E_1)$ is a tournament:
\begin{itemize}
\item if $(u,v) \in E$ but $(v,u) \not \in E$, then $(u,v) \in E_1$;
\item otherwise if $\pi(u)<\pi(v)$, we let $(v,u)\in E_1$ and $(u,v)\in E_2$.
\end{itemize}
Observe that since only the feedback arcs in $\pi$ contribute in the definitions of $\profile{S}{\pi,\sigma}$ and $\cutvector{S}{\pi}$, we have
$$(n,\profile{T}{\pi,\sigma},\cutvector{T}{\pi})=(n,\profile{S}{\pi,\sigma},\cutvector{S}{\pi}).$$
Moreover, from the assumed domination we have
$$(n,\profile{S}{\pi,\sigma},\cutvector{S}{\pi})\domleq (n',\profile{S'}{\pi',\sigma'},\cutvector{S'}{\pi'}).$$
Hence we may apply Lemma~\ref{lem:dom-tour} to get a strong immersion model $\mu$ of $T$ in $S'$ such that $\mu(\pi_j)=\pi'_{f(j)}$ for all $j\in [n]$.

Our goal is to extend $\mu$ to a strong immersion of $S$ in $S'$. 
It remains to show how for each arc of $E_2$ we can construct a corresponding directed path in $S'$ so that these paths are pairwise arc-disjoint, and also arc-disjoint with the paths used in $\mu$.
Let us partition $E_2$ into $F_1\uplus F_2$, as follows.
Take any arc $(u,v)\in E_2$, , 
and let us denote $\pi(u)=j$ and $\pi(v)=h$; recall that $j<h$. 
Recall also that since $(u,v)\in E_2$, we have $(v,u)\in E_1$.
Since the embedding $f$ is strictly increasing, we have $f(h)-f(j) \geq h-j$. 
Put $(u,v)$ into $F_1$ if $f(h)-f(j) = h-j$ and into $F_2$ if $f(h)-f(j) > h-j$.

We first construct the images for the arcs of $F_1$.
Take any $(\pi_j,\pi_h)\in F_1$.
Observe that since $f$ is strictly increasing, in fact for every $\ell \in [0,h-j]$, we have $f(j+\ell) = f(j) + \ell$. 
By domination, we have that $\codeword{S}{\pi,\sigma}(j+\ell)=\codeword{S'}{\pi',\sigma'}(f(j+\ell))=\codeword{S'}{\pi',\sigma'}(f(j)+\ell)$ for all $\ell\in [0,h-j]$.
Since the codewords contain the full information on which feedback arcs in consecutive cuts are equal, and which feedback arcs have corresponding symmetric arcs,
it can be easily seen that these equalities of the labels imply that the semi-complete digraphs $S[\pi_{\geq j} \cap \pi_{\leq h}]$ and $S'[\pi'_{\geq f(j)} \cap \pi'_{\leq f(h)}]$ are isomorphic,
with the isomorphism mapping $\pi_{j+\ell}$ to $\pi'_{f(j)+\ell}$.
In particular, we have $(\mu(\pi_j),\mu(\pi_h)) = (\pi'_{f(j)}, \pi'_{f(h)}) \in E'$. 
Observe that by the definition of a strong immersion, if $(\mu(\pi_j),\mu(\pi_h))$ belonged to some path $\mu(e)$ for an arc $e\in E_1$, then we would necessarily have $e=(\pi_j,\pi_h)$, 
however $(\pi_j,\pi_h)\in F_1$. 
Therefore the arc $(\mu(\pi_j),\mu(\pi_h))$ is free and we can set $\mu((\pi_j,\pi_h))$ to be the length-$1$ path consisting only of the arc $(\mu(\pi_j),\mu(\pi_h))$.
Observe that this preserves the invariant asserted by Lemma~\ref{lem:dom-tour} that all non-feedback arcs in $\pi$ are mapped to single-arc paths in $S'$.

We are left constructing the images for arcs of $F_2$.
Take any $(\pi_j,\pi_h)\in F_2$.
Since $j \equiv f(j) \bmod (4c+1)$ and $h \equiv f(h) \bmod (4c+1)$ by domination, we infer that $f(h)-f(j) \equiv h-j \bmod (4c+1)$. 
Therefore, $f(h)-f(j) \geq h-j+(4c+1)$. This means that there are at least $4c+1$ vertices in $\pi'_{\geq f(j)} \cap \pi'_{\leq f(h)}$ that are not images under $\mu$ of any vertex of $S$. 
Among these $4c+1$ vertices, at most $c$ can be the tails of feedback arcs with $\pi'_{f(j)}$ as the head, since each such arc contributes to $|{E'}_{\pi'}^{f(j)}|$, which is at most $c$. 
Similarly, at most $c$ of these vertices can be the heads of feedback arcs with $\pi'_{f(h)}$ as the tail. 
Since $S'$ is semi-complete, this leaves us with at least $2c+1$ indices $i$ with the following properties 
\begin{enumerate}
\item $f(j)<i<f(h)$; 
\item $\pi'_i \notin \mu(V)$; and
\item $(\pi'_{f(j)},\pi'_i)\in E'$ and $(\pi'_i,\pi'_{f(h)})\in E'$.
\end{enumerate}
Each vertex $\pi'_i$ located at such position $i$ will be called a {\em{free pivot}} for the arc $(\pi_j,\pi_h)\in F_2$.

We now verify that if $\pi'_i$ is a free pivot for $(\pi_j,\pi_h)\in F_2$, then none of the arcs $(\pi'_{f(j)},\pi'_i)$ and $(\pi'_i,\pi'_{f(h)})$ belongs to the path $\mu(e)$ for any arc $e\in E_1\cup F_1$. 
This is because:
\begin{itemize}
\item If $e$ is a feedback arc in $\pi$, then $e\in E_1$.
By the definition of strong immersion, the path $\mu(e)$ uses only at most two arcs incident to vertices of $\mu(V)$: the first and the last one on $\mu(e)$.
By the last assertion of Lemma~\ref{lem:dom-tour}, both of them are feedback arcs in $\pi'$. However, both arcs $(\pi'_{f(j)},\pi'_i)$ and $(\pi'_i,\pi'_{f(h)})$ are not feedback arcs in $\pi'$ and they
are incident to vertices of $\mu(V)$. Hence they cannot be used on $\mu(e)$.
\item Otherwise, if $e\in E_1\cup F_1$ is a non-feedback arc, then by construction $\mu(e)$ consists of a single arc connecting two vertices from $\mu(V)$. However we have $\pi'_i\notin \mu(V)$,
hence arcs $(\pi'_{f(j)},\pi'_i)$ and $(\pi'_i,\pi'_{f(h)})$ cannot be used by $\mu(e)$.
\end{itemize}
Thus, for any vertex $\pi'_i$ that is a free pivot for $(\pi_j,\pi_h)$, the path of length $2$ formed by the arcs $(\pi'_{f(j)},\pi'_i)$ and $(\pi'_i,\pi'_{f(h)})$ may be used to define $\mu(u,v)$, 
because none of its arcs has been used so far for images of arcs from $E_1\cup F_1$. Nevertheless, we have to argue that such $2$-paths can be selected so that they are pairwise arc-disjoint. 
This will conclude the construction of a strong immersion model of $S$ in $S'$.

We prove that greedily selecting $2$-paths via free pivots is safe. Iteratively consider the arcs of $F_2$, and let $e=(\pi_j,\pi_h)\in F_2$ be the next one.
Recall that since $(\pi_j,\pi_h)\in F_2\subseteq E_2$,
 we have also $(\pi_h,\pi_j)\in E$.
Observe that out of the at least $2c+1$ arcs with tail $\pi'_{f(h)}$ and head being a free pivot for $e$, only at most $c$ could have been used so far in the greedy procedure, 
since each such arc used so far corresponds to a different pair of symmetric arcs in $S$ incident to $\pi_h$; the number of such arc pairs is bounded by $c$ due to the bound on the cutwidth of $\pi$.
Symmetrically, at most $c$ arcs with head $\pi'_{f(j)}$ and tail in a free pivot for $e$ could have been used so far.
This leaves us with at least one free pivot $\pi'_i$ with both arcs $(\pi'_{f(j)},\pi'_i)$ and $(\pi'_i,\pi'_{f(h)})$ unused so far, so we can define $\mu(e)$ as the $2$-path formed 
by this pair of arcs.
\end{proof}

We conclude by formally verifying that Lemma~\ref{lem:wqo} follows from Lemma~\ref{lem:dom-semi}.

\begin{proof}[Proof of Lemma~\ref{lem:wqo}]
Let $S_1,S_2,S_3,\ldots$ be a sequence of semi-complete digraphs, each of cutwidth at most $c$.
By Lemmas~\ref{lem:lvo} and~\ref{lem:lco}, for each $t=1,2,\ldots$ we can fix a linked vertex ordering $\pi_t$ of $S_t$, and a linked sequence of ordered cuts $\sigma_t$ for $\pi_t$.
For each $t$, consider the codeword $C_t=(n_t,\codeword{S_t}{\pi_t,\sigma_t},\cutvector{S}{\pi})$, where $n_t$ is the number of vertices of $S_t$.
Since the domination order is a well-quasi-ordering of codewords, we infer that there are indices $1\leq t<t'$ such that $C_t\domleq C_{t'}$.
By Lemma~\ref{lem:dom-semi}, $S_t$ can be strongly immersed in $S_{t'}$, which concludes the proof.
\end{proof}


\bibliographystyle{abbrv}

\end{document}